%% file: main.tex
\documentclass[times, 10pt,twocolumn]{article}
\pdfoutput=1
\usepackage{latex8}
\usepackage{times}
\usepackage{listings}
\usepackage{longtable}
\usepackage{latexsym}
\usepackage{clrscode}
\usepackage{graphicx}
\usepackage{subfigure}

\usepackage[latin1]{inputenc}

\newtheorem{corollary}{Corollary}
\newtheorem{theorem}{Theorem}

\newtheorem{definition}{Definition}

\newcommand{\qed}{\hfill\leavevmode
  \hbox to.77778em{%
  \hfil\vrule
  \vbox to.675em{\hrule width.6em\vfil\hrule}%
  \vrule\hfil}}
\DeclareRobustCommand{\textsquare}{%
  \begingroup \usefont{U}{msa}{m}{n}\thr@@\endgroup
}

\newenvironment{proof}[1][\proofname]{\par
  \normalfont
  \trivlist
  \item[\hskip\labelsep
        \itshape
    #1.]
}{%
  \qed\endtrivlist
}
\providecommand{\proofname}{Proof}

\begin{document}

\title{Broadcasting in Prefix Space:\\ P2P Data Dissemination with Predictable Performance\thanks{This work is supported in part by the German Bundesministerium für Bildung und Forschung within the project  \emph{Moviecast} (http://moviecast.realmv6.org).
}
}

\author{Matthias~W{\"a}hlisch\protect\thanks{The author is also with HAW
Hamburg, Dept. Informatik, and with link-lab, Berlin.}\\ 
Freie Universit{\"at} Berlin\\ Institut f{\"u}r Informatik\\
Takustr. 9\\ D-14195 Berlin, Germany\\
waehlisch@ieee.org\\
\and Thomas~C.~Schmidt\\
HAW Hamburg\\ Dept. Informatik\\
Berliner Tor 7\\ 
D-20099 Hamburg, Germany\\
 t.schmidt@ieee.org\\
\and Georg Wittenburg\\
Freie Universit{\"at} Berlin\\ Institut f{\"u}r Informatik\\
Takustr. 9\\ D-14195 Berlin, Germany\\
wittenbu@inf.fu-berlin.de\\
}

\maketitle
\thispagestyle{empty}

\begin{abstract}
A broadcast mode may augment peer-to-peer overlay networks with an efficient, scalable data replication function, but may also give rise to a virtual link layer in VPN-type solutions. We introduce a generic, simple broadcasting mechanism that operates in the prefix space of distributed hash tables without signaling. This paper concentrates on the performance analysis of the prefix flooding scheme. Starting from simple models of recursive $k$-ary trees, we analytically derive distributions of hop counts and the replication load. Further on, extensive simulation results are presented based on an implementation within the OverSim framework.  Comparisons are drawn to Scribe, taken as a general reference model for group communication according to the shared, rendezvous-point-centered distribution paradigm. The prefix flooding scheme thereby confirmed its widely predictable performance and consistently outperformed Scribe in all metrics. Reverse path selection in overlays is identified as a major cause of performance degradation.
\end{abstract}


\input{part1}

\input{part2}

\bibliographystyle{latex8}
\bibliography{fhtw-ipv6,mipv6,imeasurement,mmodeling,ssm,fhtw-vcoip,mcast,mmcast,dht-manet,omcast,p2p,dhtbcast,vcoip,mw-unref}

\end{document}

%% file: part1.tex
\Section{Introduction}\label{sec:broadcast:introduction}
A broadcast service is commonly supported on the network and data link
layer. Analog to the IP layer, application overlays may require the use of
an unselective group communication. Distributed Hash Tables (DHT) like
Chord \cite{smkkb-cspls-01} and Pastry \cite{rd-psdol-01} do not consider
broadcast, i.e., a mechanism to communicate to all parties of one DHT
instance without their active participation.

The broadcast mode admits two unique features. The a
priori awareness of the data flooding task may significantly enhance
efficiency, e.g., by taking advantage of network or (shared) media
specifics. Further on, it enables a message exchange among mutually unknown
parties without a requirement of specific service awareness or any form of
signaling. Broadcast is thus the fundamental mechanism for unselective
data synchronization and for the autonomous coordination of distributed
systems.

On the application layer, there are likewise versatile use cases for
broadcast communication. Applications range from broadband data
dissemination in video conferencing or data replication, over 
service and peer discovery up to the implementation of a virtual link layer
in VPN-type solutions.

Broadcast is a special case of multicast. This distribution mechanism
guarantees to reach not only a subset, but all nodes of a
dedicated domain without explicit registration. The set of \emph{all nodes}
is also called the broadcast domain. It is worth noting that a broadcast
domain can be arranged on different layers with varying inherent
capabilities. Connecting nodes, e.g., with an Ethernet hub to a shared
segment facilitates packet distribution based on the physical network structure. It is limited by the supporting medium, i.e., the range of signal propagation.  The equivalent holds for the wireless domain, where
the medium is always shared, but of restrictive propagation ranges.
Participating nodes do not need a specific network logic in sending and
receiving broadcast data on the physical layer. Broadcast support, however,
on a dedicated layer should be independent of the underlying tier, which
may accelerate it. In the example of IP, broadcast addresses will be
directly mapped to the Ethernet broadcast address, such that all Ethernet
hosts of one segment receive the data independent of their subnet
assignment, but in contrast to network access, packets can be forwarded on
the network layer beyond physical bounds.

In general, broadcast in logical networks can be enabled by passing data
incrementally to direct overlay neighbors. If the graph of nodes is
connected and contains the source, all nodes will be reached. DHT
structures allow to derive such a connected neighborhood graph. Any node
can send packets to an address adjacent to its own key space. In contrast
to IP, every possible address is associated with one overlay peer. Such a
simple ring broadcast scheme sends the packet to exactly one neighbor, reaching all $n$ DHT peers after $n$ hops.
As an alternative approach to the case of unknown neighborhoods, a
dedicated, well-known replicator can be placed in the network like the
Broadcast and Unknown Server in ATM. Such a rendezvous point-based approach
requires extra signaling to register receivers. The parallelism of
distribution is bounded by the replicator, which sustains the overall
duplication load and may be a single point of failure.

In the following, we will present a general broadcast algorithm along with
optimizations for Pastry, that uses the DHT structure more efficiently and
replicates data stepwise to all neighbors in prefix space. This scheme
works without peer involvement, especially without signaling.  We model and analyze the approach theoretically and in simulation, drawing comparison to a generic rendezvous point approach derived from Scribe \cite{cdkr-slsda-02}.

This paper at first gives an introduction of the prefix flooding algorithm in the next section and continues  as follows. Section \ref{sec:bcast-performance} presents an overview of the performance measures applied in our analysis, while analytical models are utilized in section \ref{sec:bcast-analytics} to derive distributions for the core properties of replication load and hop count. Results of our simulation studies are outlined in the subsequent section \ref{sec:bcast-simresults}. Related work is reviewed in section \ref{sec:rel-work},  followed by a final discussion and conclusion  in section \ref{sec:disc-concl}.

\Section{Broadcast by Prefix Flooding}
For an efficient application layer broadcast we need to define a strategy
for data replication on the overlay. In a DHT, the peer identifiers are
composed using an alphabet of $k$ digits and have a predefined length. All
nodes of a structured overlay can be naturally arranged in a prefix tree,
branching recursively at longest common prefix of $k$ neighboring vertices.
The leaves are labeled with the overlay identifiers of the DHT members and
the inner vertices represent the shared prefix (cf. figure
\ref{fig:prefixsubgraph}).

This tree can be interpreted as a distribution tree, defining the broadcast
domain of a specific DHT instance. If a broadcast packet is sent starting
from the root of the tree towards the leaves, the packet will be replicated
where prefixes branch. Actually, the broadcast domain (prefix tree)
decomposes in many smaller broadcast sub-domains (subtrees), in which the
propagations continue in parallel. Following the nature of broadcast, a
packet will be forwarded locally, after it has arrived at a root of a
subtree.

This approach allows to reach all peers of a DHT, because the data is
flooded to the leaves, which represent the overlay nodes. A peer receiving
a broadcast is required to determine the current branching position on the
distribution tree to decide on further packet replication. This context
awareness can be achieved by sending broadcast packets carrying the prefix
currently addressed, which we call \emph{destination prefix}.  This
destination prefix will grow in length with every forwarding hop while
descending the tree. 

We denote the length of a prefix ${\cal A}$ by $|{\cal
A}|$. Given two prefixes ${\cal A}$ and ${\cal B}$, the longest common
prefix will be written ${\cal L} = LCP({\cal A}, {\cal B})$. The relation
of ${\cal L}$ being a prefix of ${\cal A}$ is written as ${\cal L}
\subseteq {\cal A}$. Consequently ${\cal L} \subseteq {\cal A}$ and ${\cal
A} \subseteq {\cal L}$ if and only if ${\cal L} = {\cal A}$.


\begin{figure}
  \center
  \includegraphics[width=0.8\columnwidth]{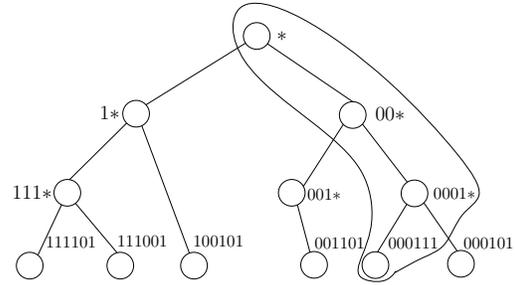}
  \caption{DHT Node  within a Prefix Tree --  \newline Associated Vertices
   are Highlighted.}\label{fig:prefixsubgraph}
\end{figure}

A proper specification for data distribution, i.e., a routing procedure on prefix trees, requires further definitions. The two sub-problems that need to be solved are \emph{a routing to a prefix} and the \emph{association of nodes with
prefixes}:
\begin{definition}
A prefix ${\cal L}$ is associated with an overlay node of ID ${\cal
N}$, if and only if ${\cal L} \subseteq {\cal N}$.
\end{definition}
As shown in figure \ref{fig:prefixsubgraph}, all inner vertices on the
shortest path from the root to a node are associated with that node.

Concordantly, a \emph{prefix routing} can be defined as forwarding a packet
to the node the destination prefix is associated with. In general, there
may be several nodes owning an associated prefix, since prefix-to-node
mapping is only assured to be unique for prefixes of full key length.  For
flooding a prefix tree, a forwarding peer needs to route packets to all
'live' neighboring prefixes (cf. figure \ref{fig:prefixsubgraph}).
Consequently, a peer must store corresponding nodes for each prefix
adjacent to its associated vertices in a prefix neighbor set. It is
important that these tables are complete. A \emph{complete neighbor set}
meets the following condition: Whenever an overlay node exists for a given
prefix, then the neighbor set will provide an entry for this prefix. In
particular it follows that each overlay node is a destination in at least
one set, since node keys are uniquely assigned. It is worth noting that a
prefix needs not to be included in any neighbor set, if there is no peer
sharing it. The requirement of complete neighbor tables will usually be
fulfilled by the key-based routing service, i.e., underlying DHT
routing maintenance.

A source initiates a broadcast by starting with the empty destination
prefix. This corresponds to delivering the data to all prefix neighbors
${\cal N}_i$. At each neighbor a packet will be further replicated. The
destination prefix is replaced with the new target address. In detail, the
algorithm works as follows:

\begin{codebox}
  \Procname{\proc{Prefix Flooding}}
  \zi \Comment On arrival of a packet with destination prefix $\cal C$ 
  \zi \Comment at a DHT node 
  \li \For all ${\cal N}_{i}$ IDs in prefix neighbor set
  \li   \Do \If ($LCP({\cal C},{\cal N}_{i}) = {\cal C}$) \>\>\>\>\>\>\> \Comment ${\cal N}_{i}$  downtree neighbor
  \li       \Then ${\cal C}_{new} \gets {\cal N}_{i}$
  \li          \proc{Forward packet to ${\cal C}_{new}$} 
  \End
\end{codebox}

If an inner vertex of the prefix tree fails, e.g., due to churn, the
corresponding sub-tree is empty or includes further peers. The replacement
of the next hop for a given prefix ${\cal C}_{new}$ will be achieved by the
underlying DHT.  In general, in the case of overlay network failures the
reliability of prefix flooding relies directly on the deployed DHT
maintenance.


If all peers have a complete set of prefix neighbors, the scheme
guarantees that all overlay nodes will be accessed, no peer receives a
broadcast packet more than once and the algorithm terminates.

\begin{theorem}[Coverage]\label{thm:bcastcoverage}
If the prefix neighbor sets are complete at all nodes, then the
\proc{Prefix Flooding} assures packet distribution to all overlay nodes.
\end{theorem}

\begin{theorem}[Uniqueness]\label{thm:bcastduplications}
Each overlay node will receive a broadcast packet at most once using the
\proc{Prefix Flooding}.
\end{theorem}

Complete proofs for both theorems are elaborated in \cite{w-sagcb-08}. Theorem  \ref{thm:bcastcoverage} can be proven by induction over the number of overlay nodes, while theorem \ref{thm:bcastduplications} follows from the observation that each routing prefix uniquely identifies the root of a subtree in prefix space.

From theorem \ref{thm:bcastduplications} it can be concluded that the \proc{Prefix
Flooding} does not induce loops, proving the assumption that the algorithm
terminates.

\SubSection{Implementation for Pastry}

The idea of prefix routing is implemented in Pastry. The Pastry routing
table of a peer reflects directly the elements of a prefix tree. Thus each peer carries a subset of the prefix tree in its routing
table.  Merging the routing tables of all peers, would form the global 
distribution tree. 
In flooding their routing tables, Pastry peers flood the prefix
tree, which corresponds to the overlay broadcast described by the
\proc{Prefix Flooding}. In detail, the idea is as follows: A source sends
its data to all routing table entries.  Each destination prefix corresponds to the root of a broadcast sub-domain. The receiving peers determine their
position in the tree, i.e., the height $D$ in the prefix tree, at which
they receive the data, and forward the packets downwards. This is equal to
sending data to all routing table entries starting at row $D+1$. Note that
the tree position can easily be derived by denoting the row number, which
reduces the packet size in contrast to encoding the entire key. For Pastry
the \proc{Prefix Flooding} reads in pseudo code:


\begin{codebox}
  \Procname{\proc{Pastry Prefix Flooding}}
  \zi \Comment On arrival of a packet with destination prefix length 
  \zi \Comment  $D$ at Pastry node of ID $\cal K$ with routing table $A$ 
  \zi \Comment containing $l$ rows and $k$ columns
  \li \For all $i \gets {D} + 1$ \To $l$
  \li   \Do \For all $j\gets 1$ \To $k$
  \li     \Do \If $a_{i,j} \neq Unspecified \wedge a_{i,j} \neq \cal{K}$
  \li       \Then $D_{new} \gets i$
  \li          \proc{Forward Packet To $a_{i,j}$}
        \End \End \End
\end{codebox}

If the routing table is filled correctly, all theorems for the \proc{Prefix
Flooding} are also valid for Pastry, since the Pastry routing table
corresponds to the set of prefix neighbors $\{{\cal N}_i\}$. However,
Pastry reactive maintenance does not guarantee that each overlay node will
provide complete routing states \cite{rd-psdol-01}, which conflicts with
the \proc{Prefix Flooding}. Therefore we augmented Pastry  with a proactive
routing maintenance mechanism, which performs initial key look-ups to fill
the routing table  similar to the ``fix\_fingers'' routine in Chord.

\Section{Performance Measures}\label{sec:bcast-performance}
The prefix flooding approach to broadcasting introduces prefix trees as a
control plane to packet forwarding. This simple mechanism operates without
additional signaling, which is an apparent advantage. The quality of the
routing as inherited from a hash-generated prefix tree needs closer
inspection. Ideally, packet distribution should be fast and minimize
traffic and replication load in the network.  To obtain an overall insight
into the routing quality, we evaluate the prefix flooding scheme in theory
and in a discrete event simulation according to the following metrics and
compare our results to Scribe \cite{cdkr-slsda-02}. Scribe serves as a
generic reference model for schemes using dedicated replicators, and is
based on the same DHT, Pastry. It is worth noting that the performance
metrics do not measure the multicast specific properties of Scribe.  Thus,
choosing Scribe for comparison is reasonable.

  {\bf\em Packet replication load} quantifies the number of packets a single
  peer has to forward. This metric reflects the number of direct neighbors
  per node in the distribution tree. The overall characteristic for the
  prefix routing is then given by the distribution of the replication load
  obtained from all forwarding nodes.

  {\bf\em Hop count} counts the number of overlay routing traversals that a
  packet needs on its way from the source to the destination. Note that
  the hop count affects the travel time, because every additional hop
  results directly in an additional transmission time. In this sense the
  travel time is correlated with the hop count.

  {\bf\em Travel time} describes the time a data packet travels from the
  source until it reaches a receiver measured in seconds. This absolute
  value depends on the one hand on the number on hops between the nodes and
  on the other hand of the transmission time inherited from the hop by hop
  link delays and the packet size of the transmitted data.

  {\bf\em Relative delay penalty} measures the ratio of the travel time for
  data packets delivered via Scribe and the travel time resulting from the
  prefix flooding scheme. This relative factor gives an indication of the
  parallelism of packet forwarding.

%% file: part2.tex
\Section{Analytical Models}
\label{sec:bcast-analytics}

To understand the performance of the prefix flooding scheme, we first
present analytical considerations. Based on the shape of the prefix tree,
we gain insight in the structural behavior of protocols for traversing
prefix distribution trees. As this analysis is only based on the tree
itself, fringe effects known from simulations are isolated. 

\SubSection{Replication Load} 
In the following, we want to derive the distribution of the replication load
in a prefix tree.
For the general case of prefix flooding in a structured overlay of $N$
nodes using a prefix alphabet of $k$ digits, the following upper bound of
the replication load can be derived immediately.

\begin{theorem}
Any overlay node in a prefix flooding domain of $N$ receivers and an
alphabet with $k \geq 2$ digits  will replicate a data packet at most
$\log_{2}(N) (k-1)$ times.
\label{thm:repli-limit}
\end{theorem}

For the distribution function of the replication load in a fully populated prefix
tree, we need to determine replication values along with their frequencies.
Recalling the picture of a full prefix tree for an alphabet with $k$
digits, every node except the leaves has $k$ children. The number of packet
replications for an overlay peer is equal to the overall number of
forwarding neighbors, which depends on the tree position, where a peer
receives the packet. Per level the replication load is $k-1$. Consequently, in a fully populated $k$-ary prefix tree of height $h$, replication occurs only at multiples of $k-1$, the number of neighbors in prefix space. For $j \geq 0$ we denote these
discrete values by $v_{h,k}(j) = (h-j)(k-1)$.

\begin{figure}
  \center
  \includegraphics[width=0.8\columnwidth]{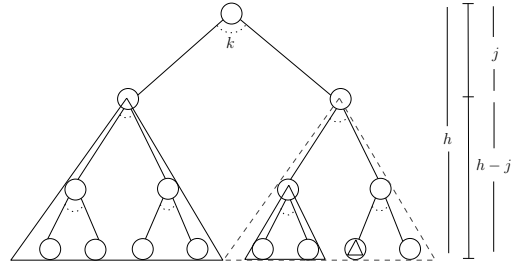}
  \caption{Self-Similarity of Prefix Subtrees due to the Recursive Nature of
  $k$-ary Trees}\label{fig:ktree-prop}
  \vspace*{-0.5cm}
\end{figure}

To derive the replication frequency, we quantify the occurrence of the
replication load $v_{h,k}(j)$. Since we know the load of a peer forwarding
packets at height $j$, the frequency can be calculated by counting the
number of peers that fulfill the replication condition. The latter
corresponds to the number of (sub-)trees with height $h-j$, because every
peer serves as forwarder for one tree. Starting at the source in a full
prefix tree, the structure decomposes in $k-1$ subtrees with height $h-1$,
$k(k-1)$ subtrees of height $h-2$, etc. (cf. figure \ref{fig:ktree-prop}).
At every level of the full prefix tree, there is an exponential growth in
the number of inner vertices representing the root of new subtrees. Thus,
the frequency of ($h-j$)-size subtrees must increase exponentially with
their decreasing height. In detail there are $k^{j-1}\cdot(k-1)$ subtrees
of height $h-j$, which account for a replication load of $(h-j)\cdot(k-1)$.

\begin{theorem}\label{thm:repli-freq}
Given a fully populated $k$-ary prefix tree of height $h$. Then the
frequency $f_{h,k}(v_{h,k}(j))$ for a replication load $v_{h,k}(j) = (h-j)(k-1)$ is
given by

\begin{equation}
f_{h,k}(v_{h,k}(j)) = \left\{ \begin{array}{ll}
                      1 & \textnormal{for } j=0 \\
                      k^{j-1}\cdot(k-1) & \textnormal{for } 0<j\leq h.
                    \end{array}
                    \right.
\label{eq:repli-freq}
\end{equation}
\end{theorem}

\begin{proof}[Proof by induction]
We assume a full $k$-ary prefix tree of height $h$. The case $j=0$
corresponds to the (single) source that replicates data to $h(k-1)$
neighbors as derived above.

The induction is done with respect to $h-j$, the height of a subtree (cf. fig. \ref{fig:ktree-prop}).

Base case: Is $h-j=1$, we have to show that the replication load
$v_{h,k}(h-1)$ appears ($k-1$)-times. In a tree of height 1, the source sends
the data to all further leaves directly, which equals  $k-1$.


Induction step: Assume the statement holds for $h-j$. We have to show that
the statement holds for $h-j+1$, i.e., 
$f_{h,k}(v_{h,k}(h-j+1)) = k^{h-j}(k-1)$. 

Consider a full prefix tree of height $h-j+1$. It 
consists of $k$ subtrees of height $h-j$. The replication load of a node in
a tree of $(h-j+1)$ equals the sum of all neighbors in these k subtrees.
Using the induction hypothesis the overall replication load reads \\ $k\cdot
f_{h,k}(v_{h,k}(h-j)) = kk^{h-j-1}(k-1) = k^{h-j}(k-1).$
\end{proof}

The overall number of packet replications is easily identified as the
number of leave nodes, since there are no packet duplications and each peer
receives the broadcast. The number of leaves of a full $k$-ary tree of
height $h$ equals $k^h$, such that we arrive at the following 

\begin{corollary}\label{cor:repli-dist}
The probability distribution $P_{h,k}$ for packet replication
multiplicities reads

\begin{equation}
P_{h,k} (v_{h,k}(j)) = \left\{ \begin{array}{ll}
                       k^{-h} & \mbox{for $j = 0$}\\
                        k^{j - h - 1} \cdot  (k - 1) & \mbox{for $1 \leq j \leq h$}\\
                       0 & \mbox{otherwise.}\\
                      \end{array}
              \right.
\label{eq:repli-dist}
\end{equation}
\end{corollary}

\begin{corollary}
The average replication load for a node in a full prefix tree $T_{h,k}$ is
given by $1 + {\cal O}(k^{-h})$, its standard deviation by $ \sqrt{k} + {\cal O}(k^{-h})$.
\label{cor:repli-avg}
\end{corollary}

Observing the weak dependence of the replication load distribution on $h$
and $k$, i.e., the tree shaping parameters, it can be assumed that the
model is sufficiently general to grant insights into the qualitative
replication behavior of a sparsely populated $k$-ary trees. We will see in section \ref{sec:bcast-simresults} that the simulations support this assumption.

\SubSection{Hop Count}
As for the replication load, we firstly derive general measures of the
number of hops a packet travels from the source to any destination in the
prefix flooding scheme. 

\begin{theorem}
Any overlay node in a structured broadcast domain of $N$ receivers and an
alphabet with $k \geq 2$ digits  will receive a packet from prefix flooding
after at most $\log_{2}(N)$ hops. In the presence of Pastry overlay
routing, the number of hops attained on average equals $\log_{2^b}(N)$ with
$k = 2^b$.
\label{thm:hop-limit}
\end{theorem}

We now want to return to considering a fully populated prefix tree and
derive the hop distribution thereof. The main idea is similar to the
replication load: A forwarding peer sends the broadcast to $k-1$ prefix
neighbors, all of them rooting an equally structured subtree of height $h-1$.
We are counting the number of paths with a length reduced by one herein.
Additionally we count the frequency of paths for the calculated hop count
in the virtual subtree containing the forwarder. This recursion results in 

\begin{theorem}
Given a fully populated $k$-ary prefix tree of height h, the frequency
$f_{h,k} (j)$ of a hop count $j$ occurring in prefix flooding is given by
\begin{equation}
f_{h,k} (j) = { h \choose j } (k-1)^j.
\label{eqn:hop-freq-full}
\end{equation}
\label{thm:hop-dist-full}
\end{theorem}

\begin{proof}
A flooding packet arriving at node $n$ after $j$ hops will admit a current
destination prefix of length $j$. Being located in a subtree of height
$h-j$, $n$ will forward the packet to its downtree neighbors, thereby
partitioning its subtree into $k-1$ further subtrees of height $h-j-1$ (cf.
figure \ref{fig:ktree-prop}). Due to the recursive nature of the $k$-ary
prefix tree, the frequency distribution satisfies the recurrence relation
 
\begin{equation}
f_{h,k} (j) = f_{h-1,k} (j) + (k-1) \cdot  f_{h-1,k} (j-1)
\label{eqn:hop-frec-full}
\end{equation}
with initial conditions $ f_{1,k}(0)=1, f_{1,k}(1)=k-1$.\\ 
Inserting $f_{h,k}$ yields the claim.
\end{proof}

This result can be interpreted in two different ways. Among all
legitimate paths in downtree routing, i.e., of length $h$, those of length
$j$ are selected and branch $k-1$ times at each of the $j$ intermediate
prefix nodes. Alternatively, flooding corresponds to a node discovery process,
where a node discovers its $v_{h,k}(j) = (h-j)(k-1)$ neighbors which in
turn discover their neighbors in the following step. Subsequent neighbor
discovery requires connect to the $j$-th part as only $(h-j)(k-1)/j$ nodes
have further neighbors.

Following a similar argument as in corollary \ref{cor:repli-dist}, it is
clear that normalization for hop count frequencies is given by $k^{h}$, the
number of leaf nodes in the full prefix tree. 

\begin{corollary}
The probability distribution $H_{h,k} (j)$ of the hop count for flooding a
full prefix tree $T_{h,k}$ evaluates to
\begin{equation}
H_{h,k} (j) = k^{-h} \cdot { h \choose j } (k-1)^j.
\label{eq:hop-dist-full}
\end{equation}
\label{cor:hop-dist-full}
\end{corollary}

\begin{corollary}
The average hop count at which a packet is received from flooding in a full
prefix tree $T_{h,k}$ is given by $<H_{h,k}> = (k-1)/k \cdot h$,  the standard deviation of the hop count distribution (\ref{eq:hop-dist-full}) equals $\sigma_{H_{h,k}} = \sqrt{(k-1) \cdot h} / k$.
\label{cor:hop-avg-full}
\end{corollary}

This average is almost independent of the prefix alphabet $k$ and can be in
some sense interpreted as the counterpart of the average replication load
as seen in corollary \ref{cor:repli-avg}. As the average number of per hop
replications is close to one, packets travel down the entire tree and reach
most of their receivers after nearly $h$ hops. The width of the hop count distribution, its standard deviation, admits a weak dependence on $k$, slowly decaying from its maximum at $k = 2$ as $k^{-1/2}$. 

In contrast to the replication load distribution, which showed only a weak
dependence on the tree shaping parameters, the hop count results strongly
depend on $h$ for the fully populated $k$-ary tree. The height $h$ is directly
related to the number of nodes $k^h$ in this tree, which does not hold for
realistic scenarios. Thus  a direct transfer to sparsely populated random
trees is questionable. 

To derive a distribution for general distribution trees, evaluations are
required on the class of all {\em random $k$-ary} trees. Unfortunately,
this turns out to be difficult. Proceeding in a significantly simpler, but
reasonable approach, we restrict the analysis to the class of {\em random recursive $k$-ary} trees
with a {\em homogeneous} probability $p$ for independent edges. In this
model, each vertex branches to each of its $k-1$ possible outdegrees independently
with probability $p$, thereby preserving the recursive nature of the fully
populated $k$-ary tree. Instead of equation \ref{eqn:hop-frec-full}, the hop
frequency of routing on this random recursive tree will be governed by the
modified rate equation

\begin{eqnarray}
f_{h,k} (j) & = & f_{h-1,k} (j) + p \cdot (k-1) \cdot  f_{h-1,k} (j-1) \ ~~~~~ \\ 
& & \mbox{with } f_{1,k}(0) = 1, f_{1,k}(1) = p(k-1).\nonumber
\label{eqn:hop-frec-p}
\end{eqnarray}

This can be solved analogously to \ref{eqn:hop-frec-full} and yields

\begin{corollary}
The probability distribution $H_{h,k}^{(p)} (j)$ of the hop count for
flooding a random recursive $k$-ary prefix tree $T_{h,k}^{(p)}$ with
homogeneous, independent edge probability $p$ evaluates to
\begin{equation}
H_{h,k}^{(p)} (j) = (1 + p(k-1))^{-h} \cdot { h \choose j } \cdot (p(k-1))^j,
\label{eq:hop-dist-p}
\end{equation}
which attains the average value $<H_{h,k}^{(p)}>~= \frac{p(k-1)}{1 + p(k-1)} \cdot h$, and the  standard deviation  $\sigma_{H_{h,k}^{(p)}}~=~\frac{\sqrt{p(k-1)\cdot h}}{1 + p(k-1)} $.
\label{cor:hop-dist-p}
\end{corollary}

The introduced edge probability $p$ is not a 'free' parameter, but a
function of the number of leaf nodes $N = (1 + p(k-1))^{h} $ in the tree.
Solving this relation for 
$ p = \frac{\sqrt[h]{N}-1}{k-1},$
and inserting typical Pastry parameters for $k= 16$, $h = 128$ and node
numbers of our simulations, will lead to the relatively small edge
probabilities, mean hop counts and standard deviations displayed in table \ref{tab:link-prob}.

\begin{table}[h]
\center
\begin{tabular}{|c||c|c|c|c|}
\hline
 &  \multicolumn{4}{|c|}{$k= 16, \ h=128$}\\
\hline
 $N$ &  10 & 100 & 1.000 & 10.000 \\
\hline
\hline
$ p$ &  0.00122 & 0.00244& 0.00370 & 0.00497\\
\hline
$<H_{h,k}^{(p)}>$ &  2.30 & 4.52 & 6.73 & 8.88 \\
\hline
$\sigma_{H_{h,k}^{(p)}} $&  1.50 & 2.09 & 2.53 & 2.87\\
\hline
\end{tabular} 
\caption{Selected Link Probabilities,  Mean Hop Counts and Standard Deviations.}
\label{tab:link-prob}
\vspace{-0.5cm}
\end{table}

These analytical results will not only support a qualitative insight into the mechanisms of prefix-based packet distribution, but will also show significant agreement with the simulation results presented in the subsequent section. 

\Section{Simulation Results}\label{sec:bcast-simresults}

In this section, we will analyze the performance of the prefix flooding
based on a stochastic discrete event simulation and compare to the behavior
of the rendezvous point-based approach Scribe.  Both, the prefix flooding
and Scribe, are implemented on top of a proactive version of the DHT substrate
Pastry. 

In detail, our simulations are performed on the network simulator platform
OMNeT++ 3.3 \cite{omnetpp}, supplemented by a preliminary version of the
overlay simulation package OverSim \cite{oversim} including Scribe and
extended by the prefix flooding implementation. Pastry has been configured
as in its original version \cite{rd-psdol-01}. Especially, we use a key
length of $128$ and an alphabet size of $16$, if not mentioned otherwise.
To investigate the scaling behavior of the protocols, the simulations are
conducted for a number of peers varying by three orders of magnitude.
None of the relative metrics described in section
\ref{sec:bcast-performance} depend on the underlay. Thus the Simple model
\cite{bhk-ofons-07} has been applied as the underlying network with a
homogeneous link delay of $1\,ms$ to analyze the network properties inside
the overlay. 

The analysis is not focusing on reliability aspects, which allows us to
neglect churn. In particular, any effects of volatile nodes would be
completely maintained by Pastry for the prefix flooding and partially for
Scribe. Rendezvous point (RP) based schemes have to reorganize the
distribution tree due to failing RPs, resulting in DHTs by new key
associations, which nevertheless is not addressed here.

Summarizing the simulation scenario, we calculate the flooding performance
on an arbitrary ($k=16$)-ary prefix tree with a fixed maximal height and a
varying number of leaves interconnected by links of identical weight. The
broadcast will be initiated by a randomly selected leaf.

\SubSection{Replication Load}
\begin{figure*}
  \center
  \subfigure[Prefix Flooding, $k=16$]{\includegraphics[width=0.44\textwidth]{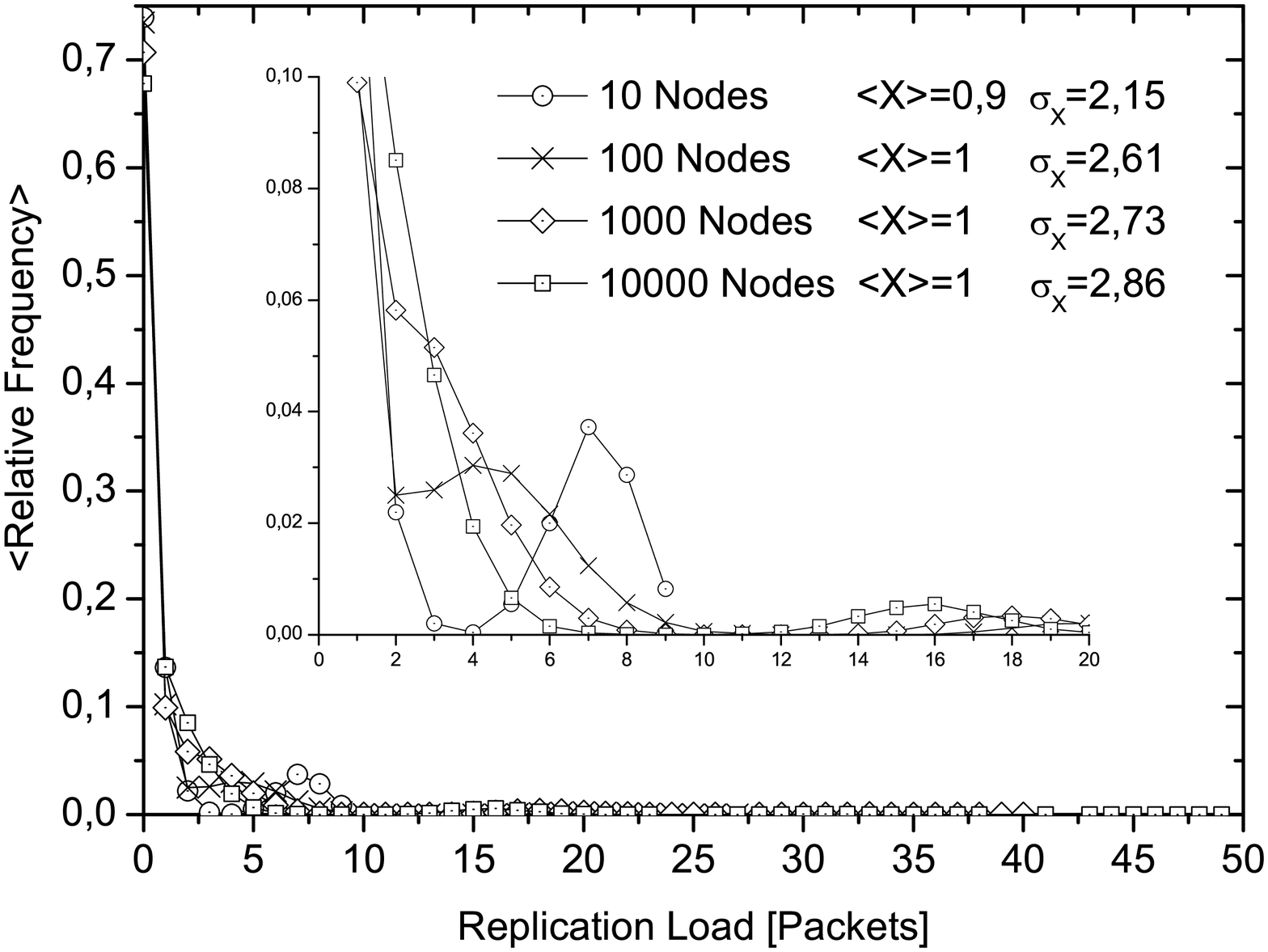}\label{fig:bcast-repliflooding}}
  \subfigure[Scribe, $k=16$]{\includegraphics[width=0.44\textwidth]{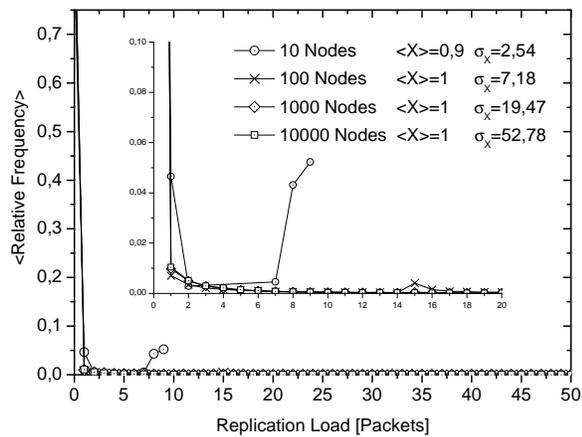}\label{fig:bcast-repliscribe}}
  \subfigure[Detail: Tail for Prefix Flooding, $k=16$]{\includegraphics[width=0.44\textwidth]{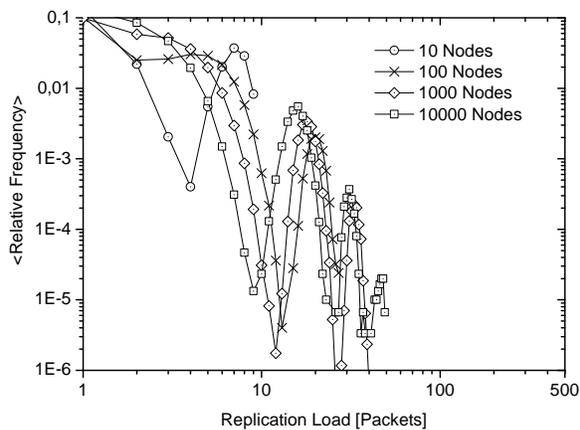}\label{fig:bcast-replifloodingtail}}
  \subfigure[Detail: Tail for Scribe, $k=16$]{\includegraphics[width=0.44\textwidth]{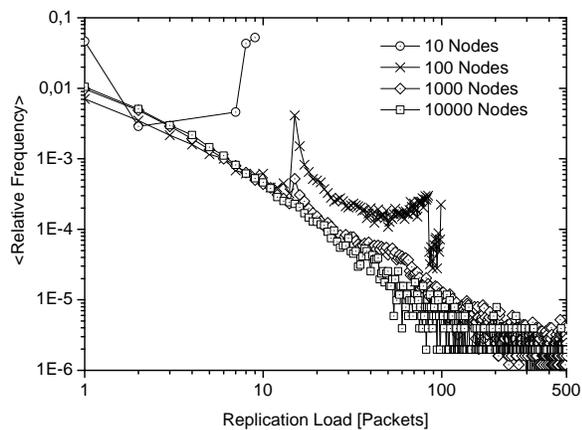}\label{fig:bcast-repliscribetail}}
  \caption{Distribution of Packet Replication Comparing Prefix Flooding
  with Scribe}\label{fig:bcast-repli}
\end{figure*}

\begin{figure*}
  \center
  \subfigure[$N=100$]{\includegraphics[width=0.44\textwidth]{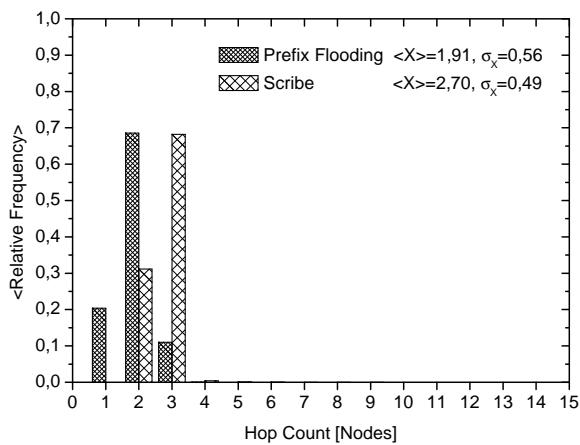}}
  \subfigure[$N=10000$]{\includegraphics[width=0.44\textwidth]{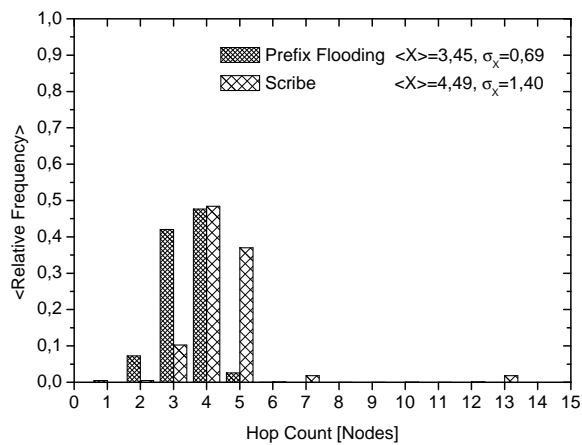}}
  \caption{Hop Count Distribution for an Overlay of Size
  $N$}\label{fig:bcast-hop count}
\end{figure*}

The distributions of the peer replication load for prefix flooding and Scribe
are displayed in figure \ref{fig:bcast-repli}. Both schemes show an exponential decay around their common average value of 1. However, the shapes of the distributions for the two approaches vary significantly, which becomes apparent at first from standard deviation values. While the widths of the distributions for prefix flooding are small und almost independent of network sizes, the corresponding values for Scribe grow large, about linearly in the number of nodes.

Both broadcasting schemes produce a large number of replications of values $0$ and
$1$, but frequencies drastically drop for higher multiplicities. Prefix flooding distribution attains a much smoother decay, leaving significant probability to replication values of $2-10$. Smoothness is even more pronounced for smaller alphabets, which for space restrictions are omitted here. 
 In contrast, Scribe decreases faster from its
average, decaying rapidly to probabilities below $1/100$ for replications larger than 2, fairly independent of the alphabet $k$. 

An exception from this overall shape can be observed for the distribution of 10 peers in Scribe. Here, the frequencies of replication values around $9$ are strongly enhanced. This border effect for  very small networks can be understood from analyzing distribution tails. As visualized in the log-log plot \ref{fig:bcast-repliscribetail}, the distribution of Scribe is heavy-tailed according to a power law decay, representing remarkably high probabilities for very large replication values up to $7800$. Corresponding probabilities are  accumulated for small sized overlays.

In contrast, the prefix flooding distribution admits a strict exponential
decay, with tail weights vanishing above 50.  Replication values in prefix
flooding are superimposed by oscillating frequencies as visible in figure
\ref{fig:bcast-replifloodingtail}. The resulting probability ``bumps'' are
noticeable on different scales for all overlays and can be explained by our
theoretical analysis, which reveals an exponential decay within the range
of multiples of $(k-1)$. Compared to the prerequisites  of corollary
\ref{cor:repli-dist}, the simulated overlays do not operate on full $k$-ary
prefix trees. Hence replication values do not only occur as multiples of
the branching factor, but level out with neighboring values. Nevertheless,
regarding the peaks of the bumps, the population and replication pattern of
the $k$-ary trees remain clearly visible. 

In both approaches, most of the peers receive the broadcast without a need to
forward it further. Scribe thereby stresses a small number of peers to serve a
much higher replication load. Instead, the prefix flooding reduces the maximal replication load by distributing the load  evenly over the neighbors. 

\SubSection{Hop Count}

The mean hop count distribution for different overlay sizes is shown in
figure \ref{fig:bcast-hop count}. In general, both schemes show the
logarithmically growing hop path length dependent on the number of peers.
With an increasing quantity of leaves, the height of prefix trees will
increase logarithmically, as well, resulting in longer paths from the
source and intermediate forwarders to the receivers. The mean hop count $<X>$ for
Scribe highlights approximately one additional node in contrast to the
prefix flooding.

For a sufficiently large $N>10$, the average of the distribution for the
prefix flooding attains directly the calculated mean hop count in theorem
\ref{thm:hop-limit}, at which all other hop count values are centered. The
hop count distribution in Scribe shows a heavy-tailed behavior, which
increases with the overlay size as indicated by the approximate linear
growth of the standard deviation. In contrast, the prefix flooding almost
attains a constant variation. Consequently, in prefix flooding the path
lengths are tightly concentrated around the logarithmically bounded
average, while Scribe builds up longer branches with higher weights.

\SubSection{Relative Delay Penalty}
Figure \ref{fig:bcast-travelpenalty} shows the 
relative delay penalty (RDP) as function of the network size for Scribe over prefix
flooding. Scribe packets travel about a factor of $1.4$ slower  than data of prefix flooding in larger networks. The enhanced delay penalty in small networks of about $10$ peers reflects the observations of figure \ref{fig:bcast-repliscribe}   that  almost all receivers are addressed directly by  the rendezvous point, which replicates  the full number of overlay nodes. The more keys are allocated, the more
branching points are located close to the RP resulting in longer paths
and less efficient parallelism in Scribe, which is in contrast to the prefix flooding.

\begin{figure}
  \center
 \includegraphics[width=\columnwidth]{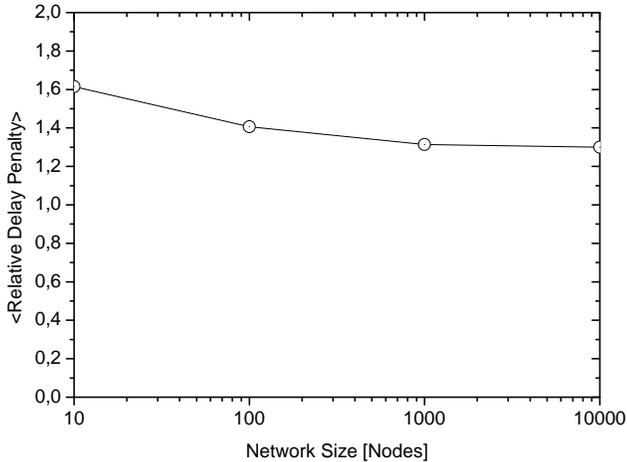}
  \caption{RDP for Scribe over Prefix Flooding}\label{fig:bcast-travelpenalty}
  \vspace*{-0.5cm}
\end{figure}

\Section{Related Work}\label{sec:rel-work}

The principal approach for implementing broadcast on a pure DHT derives from
recursive partitioning of the key space with data distribution following partition
ranges. The prefix flooding operates in this sense, defining numerical interval
boundaries from prefix transitions. The first idea of a broadcast based on
nested intervals was proposed in \cite{eabh-ebspn-03}. The broadcast is sent to intervals of exponentially increasing scale as derived from the Chord routing table. 

A  generalization  of \cite{eabh-ebspn-03} is proposed in
\cite{gaebh-sbdht-03}. In addition to a design independent of Chord, the authors
enhance their algorithm by reliability routines, which guarantee a
broadcast distribution independent of the routing table states. This is performed by
delegating data delivery for missing entries to subsequent forwarders.

The authors in \cite{lczll-ebadh-07} introduce a scheme, which splits the
key space in $d$ partitions of equal size and selects the first node in
clockwise direction as the responsible forwarder. Otherwise similar to
\cite{eabh-ebspn-03}, this approach refrains from  using uneven, logarithmic partitioning. 

An approach, which cannot ensure a broadcast distribution without data
redundancy, is presented in \cite{mg-ebpg-05}. The authors combine a
slightly enhanced version of \cite{eabh-ebspn-03} with an epidemic
distribution. All broadcast forwarders send the data periodically to
a randomly chosen neighbor, whereby the protocol may duplicate broadcast to the same 
neighbor. 
All of the approaches mentioned above lack formal verification, as well
as  analytical considerations regarding data distribution in 
$k$-ary prefix trees. Most of the algorithms are implemented on top of Chord, none of them on Pastry, which natively offers a proximity-aware prefix routing.

A generalized construction scheme to partitioning the key is space is
presented in \cite{lsl-iabdh-05}. The authors observe that any contractive
self-mapping function $P$ of the key space with a single fixed point
$\alpha$, i.e., $P(\alpha)=\alpha$, gives rise to a parent relationship.
Based on the parent relation $P(\alpha)$, a reverse path can be set up for
any node $\alpha$, leading to a broadcast distribution tree with the root
$\alpha$.  Different parent functions thus give rise to different trees at
variable roots, which may be used for load-sharing or redundancy purposes.



DHT specific flooding has been introduced in the early work
\cite{rhks-amucn-01} for CAN (Content Addressable Network). In contrast to
Chord or Pastry, CAN maps node IDs to regions representing coordinates in a
partitioned $d$-dimensional space. CAN broadcasts the data to all
geographical neighbors, thereby accounting for predecessors and foreseeable
redundancies. However, the partitioning of the $d$-dimensional space may be
uneven and result in data duplication at sub-regions. Performance
properties of multicast on CAN are derived analytically in
\cite{ws-kbrbm-09}. An extensive simulation study of flooding and tree
based overlay multicast over CAN and Pastry with respect to the underlay is
presented in \cite{cjkrt-esalm-03}.  The authors show that CAN flooding is
outperformed by Pastry flooding, which relies on a more efficient tree
structure adaptive to the underlay.

Our implementation of the generalized prefix flooding is similar to the
Pastry flooding of Castro \emph{et al.} \cite{cjkrt-esalm-03}. The main
difference lies in the reactive routing maintenance, which may result in
data redundancy at the fallback forwarder \cite{cjkrt-esalm-03}. The focus
of their analysis of broadcast distribution lies in the context of overlay
multicast. Results are only based on simulations. The measured metrics
reflect performance issues focusing on efforts imposed on the underlying
network. In this sense, our work can be understood as complementary: We
presented a general prefix flooding and investigate its inherent, structural
properties using an analytical model \emph{and} simulations.

\Section{Discussion and Conclusions} \label{sec:disc-concl}

In this work, we have presented and analyzed broadcasting within distributed
hash tables. A general prefix flooding approach, distributing data along
prefix branches directly to receivers, is compared to  a rendezvous
point-based scheme which utilizes a shared tree rooted at a predefined
anchor peer. Several phenomena of general interest could be observed.

{\bf\em Divergent Path Length Distributions:}
Our simulation results confirm the mean hop difference of one between the
prefix flooding and the rendezvous point-based approach Scribe. This
additional, triangular hop in the overlay becomes noteworthy when stretched in the
underlay and then may put stress on several links. The major advantage of the prefix flooding, though, is its quite
stable concentration of path length distribution around the average, attaining low
variations independent of the overlay size. In general, P2P networks
consist of volatile nodes. If we assume an overlay with regular churn,
i.e., session times in the range of minutes or larger, and a persistent
number of peers on average, the DHT moderately reorganizes key associations. Such
structural modifications lead to changing paths within the overlay and in the
worst case, a single arrival or departure of a node may cause a data path
to change drastically. In the prefix flooding, the path length only changes
moderately for new and existing peers due to its narrow distribution.
In contrast, the heavy-tailed overlay hop count
distribution of Scribe produces a largely inhomogeneous travel time, which
complicates synchronous applications.

{\bf\em Varying Replication Load:}
A high variation can also be identified for the packet replication in Scribe.
Similar to the prefix flooding, it is rather likely that peers forward with low
replication load. Nevertheless, in a long tail distribution  nodes are required
to replicate many more packets with values up to $7.800$ in large sized overlays
of $10.000$ peers. The distribution of packet replication is thus strongly
unbalanced, requiring very low and very high values to be served within the
same scenario. Such behavior does not only degrade the performance, but may threaten stability and  even cause conflicts with intrusion detection systems. 

In contrast to Scribe, the prefix flooding guarantees a replication load closely
balanced around its  average of about 1. It can be tuned directly
by the branching factor $k$. As we know from the theoretical analysis of
section \ref{sec:bcast-analytics}, packet replications occur as multiples of
$k-1$ in full prefix space. Decreasing $k$ adjusts the maximum number of
replications to smaller values.

{\bf\em An Overloaded Single Peer:}
The peers with extraordinarily high packet replication load in Scribe have been
identified as the rendezvous points (RP). An appropriate treatment of such service nodes becomes more important under the aspect of unbalanced packet replication, but poses a severe conceptual problem in DHTs: The placement of this entity should account for node and network capacities, but in a DHT is bound to the structural mapping of the multicast group identifier to an overlay key.
Any alternative approach, e.g., selecting the RP address  independently of the group address, will break the key space semantic with the result that an overlay node cannot derive the RP distribution address automatically.

Our prefix-guided broadcast strictly adheres to forward-directed
establishment of distribution trees. We have shown the generation of
efficient group communication structures. The presented approach is thus
particularly promising for overlay multicast services. Having sketched a
structured multicast solution operating in prefix-space \cite{ws-buode-07},
its elaboration is subject to our currently ongoing work. Further on, we
will integrate our scheme in hybrid group communication architectures
\cite{wsw-ggcns-09}.